\documentclass[11pt,a4paper]{article}

\usepackage{graphicx}
\graphicspath{{figures/}}
\usepackage[english]{babel}
\usepackage{latexsym}
\usepackage{url}
\usepackage{amsmath}
\usepackage{amsthm}
\usepackage{dsfont}
\usepackage{ifthen}
\usepackage{fancybox}
\usepackage{microtype,stmaryrd}
\usepackage{xspace}

\usepackage[charter]{mathdesign}

\usepackage{color}
\definecolor{blu3}{rgb}{.1,.0,.4}
\usepackage{hyperref}
\hypersetup{colorlinks=true, linkcolor=blu3, urlcolor=blu3, citecolor=blu3, pdfpagemode=UseNone, pdfstartview=} %colored links, no rectangles, no bookmarks, zoom

\usepackage[margin=1.1in]{geometry}

\newtheorem{theorem}{Theorem}

\def\DEF#1{\textbf{\emph{#1}}}

\newcommand{\mcu}{\textsc{Global-Mixed-Cut}\xspace}
\newcommand{\stmcu}{\textsc{Rooted-Mixed-Cut}\xspace}
\newcommand{\bmkvc}{\textsc{Bipartite Maximum $k$-Vertex Cover}\xspace}
\newcommand{\kcli}{\textsc{$k$-Clique}\xspace}

\newcommand{\probdef}[3]{
\begin{quote}
	#1\\
	Input: #2\\
	Question: #3
\end{quote}
}

\begin{document}

\title{The Complexity of Mixed-Connectivity}

\author{{\'E}douard Bonnet\thanks{Univ Lyon, CNRS, ENS de Lyon, Universit\'e Claude Bernard Lyon 1, LIP UMR5668, France. Email address: \texttt{edouard.bonnet@ens-lyon.fr}.}
\and
	Sergio Cabello\thanks{Faculty of Mathematics and Physics, University of Ljubljana, Slovenia, and 
			Institute of Mathematics, Physics and Mechanics, Slovenia.
            Supported by the Slovenian Research Agency, program P1-0297 and projects J1-9109, J1-1693, J1-2452.
			Email address: \texttt{sergio.cabello@fmf.uni-lj.si}.}
}

\maketitle

\begin{abstract}
  We investigate the parameterized complexity in $a$ and $b$ of determining whether a graph~$G$ has a subset of $a$ vertices and $b$ edges whose removal disconnects $G$, or disconnects two prescribed vertices $s, t \in V(G)$.
  
\end{abstract}

%%%%%%%%%%%%%%%%%%%%%%%%%%%%%%%%%%%%%%%%%%%%%%%%%%%%%%%%%%%%%%%%%%%%%%%%%%%%%%%%%%%%%%%%%%%%%%%%%%%%%%%%%%%%%%%%%%%%%%

\section{Introduction}\label{sec:intro}

Vertex- and edge-connectivity are fundamental concepts in graph theory and combinatorial optimization. 
They provide a basic measure of the vulnerability of a network with respect to failures, and serve as building blocks or lie at the heart of several more advanced concepts (such as flows, well-linkedness, and expanders).
A graph $G$ with at least $k+1$ vertices is \DEF{$k$-vertex-connected} if the removal of any $k-1$ vertices of $G$ leaves a connected graph.
Similarly, a graph $G$ is \DEF{$k$-edge-connected} if no removal of $k-1$ edges can disconnect~$G$.

It is also very common to consider the \DEF{rooted} connectivity, or $(s,t)$-connectivity.
For rooted vertex-connectivity, we are also given two non-adjacent, distinct vertices $s,t$ of the graph $G$, the roots, 
and we ask whether the removal of 
any $k-1$ vertices distinct from $s$ and $t$ leave the roots $s$ and $t$ in the same connected component.
For rooted edge-connectivity, the vertices $s,t$ are arbitrary, 
meaning that the edge $st$ may be present in the graph $G$,
and ask whether the removal of any $k-1$ edges leaves some path from $s$ to $t$.

To make the distinction clear, we talk about \emph{rooted} connectivity when we want to disconnect two prescribed vertices and about \emph{global} connectivity when we want to obtain (at least) two connected components.

An alternative interpretation of the rooted connectivity is through hitting sets of the $s$-$t$ paths of the graph.
This connection is made explicit by Menger's theorems, that relate the rooted 
vertex-connectivity to the number of internally vertex-disjoint paths from $s$ to $t$ and the 
edge-connectivity to the number of edge-disjoint paths from $s$ to $t$.
Since the number of vertex and edge disjoint paths can be computed in polynomial time using algorithms for maximum flow,
we can compute the rooted and the global vertex- and edge-connectivity of a graph in polynomial time.

Beineke and Harary~\cite{beineke_harary_1967} considered a natural version of the rooted connectivity where vertices
and edges are removed simultaneously and claimed a Menger-like theorem combining vertex and edge-disjoint paths.
For integers $a,b$, an \DEF{$(a,b)$-mixed cut} is a pair $(W,F)$ such that $W\subset V(G)$,
$F\subset E(G)$, $|W|\le a$, $|F|\le b$ and $(G-F)-W$ is disconnected.
For the rooted version we define a rooted \DEF{$(a,b)$-mixed cut for $s$ and $t$} to be a pair $(W,F)$ 
such that $W\subset V(G)\setminus\{s,t \}$, $F\subset E(G)$, $|W|\le a$, $|F|\le b$ and in $G-(W\cup F)$ there is not path from $s$ to $t$.
The $k$-vertex-connectivity is equivalent to the lack of $(k-1,0)$-mixed cuts, while the $k$-edge-connectivity is equivalent to the lack of $(0,k-1)$-mixed cuts (possibly with respect to roots $s, t$).

The claim of Beineke and Harary was that, if $G$ has no rooted $(a-1,b)$-mixed cut and no rooted $(a,b-1)$-mixed cut for $s$ and $t$, then there exists $a+b$ edge-disjoint paths between $s$ and $t$, of which $a$ are internally pairwise disjoint. 
Note that, contrary to Menger's theorem, the implication is claimed only in one direction, and thus it does not provide a characterization.
Nevertheless, Mader~\cite{mader_1979} pointed out that the proof in~\cite{beineke_harary_1967} is not satisfactory, and the truth of the claim is currently unclear.
The problem has been recently revisited by Johann et al.~\cite{1908.11621} for small values of~$b$ as well as for graphs of treewidth $3$.
The behavior of mixed connectivity of Cartesian product of graphs was considered by Erve{\v s} and {\v Z}erovnik~\cite{ErvesZ16}.

Sadeghi and Fan~\cite{SadeghiF19} claimed that $G$ has no 
$(a,b-1)$-mixed cut if and only if $G$ is $(a+1)$-vertex-connected and $(a+b)$-edge-connected.
While the forward direction is simple, the reverse direction of the implication is wrong and has been retracted by the authors.
As Figure~\ref{fig:counterexample} shows, the latter direction cannot be corrected if we replace the property of $(a+b)$-edge-connectivity 
by $\ell$-edge-connectivity for any $\ell$ depending on $a$ and $b$.
As a consequence, since some of the results in ~\cite{GuSF19} are using this erroneous characterization, their validity is unclear.

\begin{figure}
\centering
	\includegraphics[page=5,scale=.9]{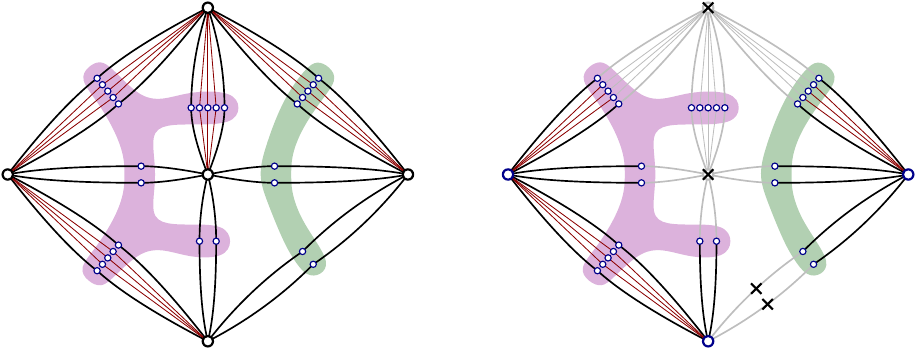}
	\caption{Example of a $3$-vertex-connected graph with arbitrarily large edge-connectivity (in particular, edge-connectivity $2+3=5$), but yet admitting a $(2,2)$-mixed cut.
		Vertices inside each single shaded region form a clique; edges of those cliques are not drawn to keep readability.
		(In the example, the graph is $5$-edge-connected, as we have a ``spanning-tree-like of thickness $5$".)}
	\label{fig:counterexample}
\end{figure}

Our focus in this paper is to analyze the computational complexity of deciding 
whether a graph has a global $(a,b)$-mixed cut or a rooted $(a,b)$-mixed cut, when parameterized by $a$ and/or~$b$. 
More precisely, we consider the following computational problems.

\probdef{\mcu}{An undirected graph $G$, and two positive integers $a$ and $b$.}{Can $G$ be disconnected by the removal of at most $a$ vertices in $V(G)$ and at most $b$ edges in $E(G)$?}

\probdef{\stmcu}{An undirected graph $G$, two distinct vertices $s, t \in V(G)$, and two positive integers $a$ and $b$.}{Can the removal of at most $a$ vertices in $V(G) \setminus \{s,t\}$ and at most $b$ edges in $E(G)$ leave $s$ and $t$ in two distinct connected components?}

These problems can also be stated equivalently as connectivity questions, where one has to be careful
on how to define mixed-connectivity.

The central focus in parameterized complexity~\cite{CyganFKLMPPS15,df13} is whether problems can be solved in time $f(k)n^{O(1)}$, where $n$ is the input size, and $k$ is a parameter value of the instance.
Algorithms with such a running time 
are called fixed-parameter tractable, or FPT for short, in the parameter~$k$.
In our case, we have two natural parameters: $a$ and $b$.
We wonder if these mixed-cut problems on $n$-vertex graphs can be solved in time $f(a)n^d$, $g(b)n^d$ or $h(a,b)n^d$, for some functions $f(\cdot)$, $g(\cdot)$, $h(\cdot)$ and some constant $d$.

\paragraph{Previous results.}
The following problem will be relevant for the forthcoming discussion.
\probdef{\bmkvc (or \textsc{Bipartite Partial Cover})}{An undirected bipartite graph $G$, two positive integers $k$ and $p$.}{Are there $k$ vertices in $V(G)$ touching at least $p$ edges in $E(G)$?}
This problem generalizes the classic \textsc{Vertex Cover} problem on bipartite graphs, by setting $p = |E(G)|$.
However it turns out to be a difficult problem, unlike \textsc{Bipartite Vertex Cover}.

Using the fact that \bmkvc is NP-hard \cite{Apollonio14,Caskurlu17,Joret15}, Rai et al.~\cite{Rai16} and 
Johann et al.~\cite{1908.11621} have noted that that \stmcu is NP-hard.
The basic idea is to attach the vertex $s$ to every vertex of one side of the bipartition and the vertex $t$ to every vertex on the other side.
Now disconnecting $s$ and $t$ by removing $k$ vertices and at most $|E(G)|-p$ edges is equivalent to finding in the bipartite graph $k$ vertices covering at least $p$ edges, which is precisely \bmkvc.

Note that this reduction does \emph{not} imply NP-hardness for \mcu; in the constructed graph, a global mixed-cut could very well disconnect a different pair than $s$ and $t$.
We observe also that \bmkvc is known to be FPT in $k$ \cite{Amini11} and in $|E(G)|-p$ (i.e., number of edges not touched by the $k$ vertices).
Therefore their reduction does not imply parameterized hardness by $a$ only nor by $b$ only.

In the same paper by Rai et al.~\cite{Rai16} it is shown that \stmcu, and even a far-reaching generalization of it, is fixed-parameter tractable (FPT) in $a$ and $b$ combined.
They develop a self-contained algorithm running in time $2^{O((a+b)^3\log{(a+b)})} n^4 \log n$.

The problem can be interpreted as an optimization problem: remove $a$ vertices and
minimize the edge-connectivity (or rooted edge-connectivity) of the remaining graph.
This problem, and generalizations of it,
have been considered in the context approximation algorithms; 
see~\cite{ChuzhoyMVZ16} and references therein.

\paragraph{Our contribution.}
In Section~\ref{sec:hardness} we show that \mcu is in fact also NP-complete.
Actually we show that \mcu, and hence \stmcu, are even W[1]-hard parameterized by $b$ only (i.e., the maximum number of edges to remove).
We also prove that \stmcu is W[1]-hard parameterized by $a$ only (i.e., the maximum number of vertices to remove).

As noted before, Rai et al.~\cite{Rai16} show that \stmcu is fixed-parameter tractable in $a+b$ with a running time
of $2^{O((a+b)^3\log{(a+b)})} n^4 \log n$ for graphs with $n$ vertices.
One may wonder whether the known heavy machinery, that one could summarize as ``small treewidth or large clique minor or large flat wall'', used for instance to solve $k$-Disjoint Paths in cubic \cite{gm13} and then quadratic time \cite{KKR12},  
can also solve \stmcu in quadratic time. 
In Section~\ref{sec:irrelevant} we show that a straightforward application of the technique does not work;
a bottleneck is the case of large clique minor.
This of course does not exclude the option for faster algorithms modifying the approach.

%%%%%%%%%%%%%%%%%%%%%%%%%%%%%%%%%%%%%%%%%%%%%%%%%%%%%%%%%%%%%%%%%%%%%%%%%%%%%%%%%%%%%%%%%%%%%%%%%%%%%%%%%%%%%%%%%%%%%%%%%%%%%%%%%%%%%%%%%%%%%%%%%%%%%%%%%%%%%%%%%%%%%%%%%%%%%%%%%%%%%%%%%%%%%%%%%%%%%%%%
\section{Preliminaries and notation}
For a graph $G$ and a subset $S$ of its vertices, $G[S]$ is the subgraph of $G$ induced by $S$.
Thus, $G[S]=(S, \{uv \in E(G)\mid u,v \in S \})$.
For a graph $G$ and two disjoint subsets of vertices $X, Y \subseteq V(G)$, we denote by $E_G(X,Y)$ the set of edges with one endpoint in $X$ and another endpoint in $Y$.
Thus, $E_G(X,Y)=\{ xy\in E(G)\mid x\in X,~ y\in Y\}$.

We provide a quick, informal overview of the concepts we will use from parameterized complexity and refer the interested reader to the standard textbooks, such as~\cite{CyganFKLMPPS15,df13}, for a comprehensive treatment.

In the \kcli problem, given a graph, one is asked whether it contains a clique of size~$k$, that is, a subset of~$k$ vertices with all the edges between them.
The \kcli problem is a $W[1]$-complete problem, hence unlikely to have an FPT algorithm; 
see~\cite[Theorem 21.2.4]{df13} or~\cite[Chapter 13]{CyganFKLMPPS15} for statements of this classical result.
The inputs of parameterized problems are pairs, formed by an instance $I$ and a parameter value $\kappa(I)$, related to a feature of the instance other than its size.
The most natural parameters are the size of the desired solution or integer thresholds used in the problem definition.

Consider a parameterized problem $\Pi$ with parameter $\kappa$.
In an fpt-reduction from \kcli to~$\Pi$, we reduce a \kcli-instance $(G,k)$ to a $\Pi$-instance $(I,\kappa(I))$ such that $\kappa(I)$ depends \emph{only} on $k$, not on the size of $G$.
An fpt-reduction from \kcli to $\Pi$ shows that $\Pi$ is W[1]-hard with respect to the parameter $\kappa$.
The intuition is that, if we would be able to solve the problems of $\Pi$ in time $f(\kappa(I)) \cdot p(|I|)$ for some function $f(\cdot)$ and some polynomial $p$, then we could solve \kcli in time $g(k)\cdot q(n)$ for a function $g(\cdot)$ and a polynomial $q$.

Another cornerstone is the Exponential Time Hypothesis (ETH); for its precise definition we refer to the textbooks. 
One of the important consequences of the ETH, which is potentially weaker than the ETH, 
is that a SAT problem with $n$ variables and $m$ clauses cannot be solved
in time $2^{o(n)}p(n,m)$ for any polynomial $p(\cdot,\cdot)$.
Assuming the ETH, there is no algorithm to solve the \kcli problem in $f(k)n^{o(k)}$ time for any computable
function $f(\cdot)$;
see~\cite[Theorem 29.7.1]{df13} or~\cite[Theorem 14.21]{CyganFKLMPPS15}.

Assume that we have an fpt-reduction from \kcli with parameter $k$ to instances $I$ of $\Pi$ with parameter $\kappa$ such that $\kappa(I)=O(k)$.
Under the ETH, we can conclude that the instances $I$ of $\Pi$ cannot be solved in time $g(k)|I|^{o(\kappa)}$ for any computable function $g(\cdot)$.
Otherwise, we could use the reduction to solve the \kcli problem in $g(O(k)) n^{o(O(k))}$, which would contradict the ETH.

%%%%%%%%%%%%%%%%%%%%%%%%%%%%%%%%%%%%%%%%%%%%%%%%%%%%%%%%%%%%%%%%%%%%%%%%%%%%%%%%%%%%%%%%%%%%%%%%%%%%%%%%%%%%%%%%%%%%%%%%%%%%%%%%%%%%%%%%%%%%%%%%%%%%%%%%%%%%%%%%%%%%%%%%%%%%%%%%%%%%%%%%%%%%%%%%%%%%%%%%
\section{Parameterized hardness with respect to $a$ only or $b$ only}
\label{sec:hardness}

\begin{theorem}\label{th:param-by-a}
  \stmcu is $W[1]$-hard parameterized by $a$ only.
  Moreover, unless the ETH fails, there is no computable function $f$ such that \stmcu can be solved in time $f(a)|V(H)|^{o(a)}$ on instances $(H,s,t,a,b)$.  
\end{theorem}
\begin{proof}
  We reduce from the \kcli problem, which is W[1]-complete parameterized by the solution size~$k$;
  see the discussion above.
  Let $(G,k)$ be an instance of \kcli.
  We build an equivalent instance $(H,s,t,a:=k,b:=|E(G)|-\binom{k}{2}$)  of \stmcu in the following way.
  See Figure~\ref{fig:param-by-a-1} for an example.
  Let $V=V(G)$, $E=E(G)$ and $m=|E|$.

  We start the description of $H$ with the vertex $s$ that we make adjacent to a clique $C$ of size $a+b+1$.
  We add to $H$ all the vertices $V$, without any edges between them,
  and make $C$ fully adjacent to each vertex of $V$.
  We add to $H$ an independent set $Z_E$ in one-to-one correspondence with the edges of $G$.
  We denote by $z_e$ the vertex corresponding to the edge $e \in E(G)$,
  and we link $z_e \in Z_E$ to $v \in V$ whenever $v$ is an endpoint of $e$.
  We finally add the new vertex $t$ that we fully link to $Z_E$.
  To summarize, $V(H) := \{s\} \cup C \cup V \cup Z_E \cup \{t\}$, 
  and the edges of $H$ can be described as: the clique $C$ is fully adjacent to the independent set $V(G)\cup \{s\}$, 
  $t$ is fully adjacent to $Z_E$, and $E_H(V,Z_E)$ is (isomorphic to) the vertex-edge incidence graph of $G$.
  We allow to delete up to $a:=k$ vertices and $b:=m-\binom{k}{2}$ edges.
  We now show the correctness of the reduction: the graph $G$ has a $k$-clique if and only if the graph $H$ has an $(a,b)$-mixed cut for $s$ and $t$.

  \begin{figure}
  \centering
	\includegraphics[page=1,scale=1]{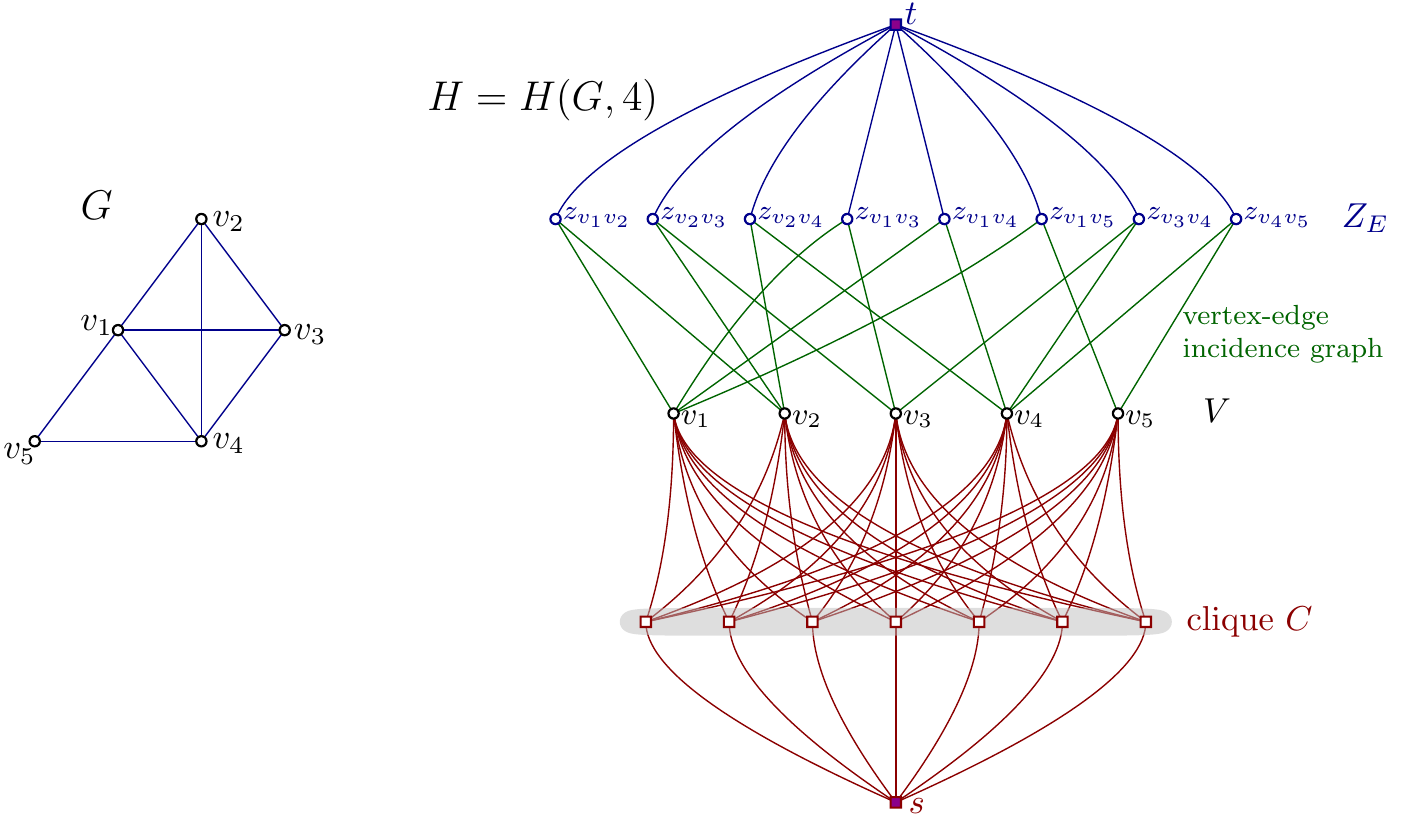}
	\caption{Example showing the reduction in the proof of Theorem~\ref{th:param-by-a}.
		On the left side we have an instance $(G,4)$ for the problem \kcli, and on the right
		we have the instance $(H,s,t,a,b)$ with $a=4$ and $b=8-6=2$ for \stmcu.
		All the vertices in the shaded region form a clique.}
		\label{fig:param-by-a-1}
  \end{figure}
  
  Let us assume that $G$ admits a $k$-clique $S \subset V$. 
  See Figure~\ref{fig:param-by-a-2} to see the construction in the example of Figure~\ref{fig:param-by-a-1}.
  Let $Z' \subset Z_E$ be the set of vertices $z_e\in Z_E$ 
  such that $e \in E(G)$ has at least one endpoint outside $S$,
  and let $F \subseteq E(H)$ be all the edges between $t$ and $Z'$.
  We claim that $(S,F)$ is an $(a,b)$-mixed cut for $s$ and $t$, hence a solution for \stmcu.
  The set $S$ is indeed of size $a=k$, and the number of edges of $F$ is $|Z'|=m-e(S)$, 
  where $e(S)$ is the number of edges in $G[S]$.
  Since $S$ is a $k$-clique in $G$, we have $e(S) = \binom{k}{2}$ and thus $|F|=m-\binom{k}{2}=b$.
  It remains only to argue that there is no path between $s$ and $t$ in $H' := (H-F)-S$.
  The only vertices in $H'$ adjacent to $t$ are the 
  vertices $z_e\in Z_E$ for which $e$ is an edge of the clique induced by $S$, 
  namely the vertices $Z_S := Z_E \setminus Z'$.
  On the other hand, since $N_{H}(Z_S) = \{ t \} \cup S$, 
  in the graph $H'$ the vertices $Z_S$ are only adjacent to $t$.
  We conclude that 
  $\{t\} \cup Z_S$ is a (maximal) connected component in $H'$,
  and therefore there is no $s$-$t$ path in $H'$.
 
  \begin{figure}
  \centering
	\includegraphics[page=2,scale=1]{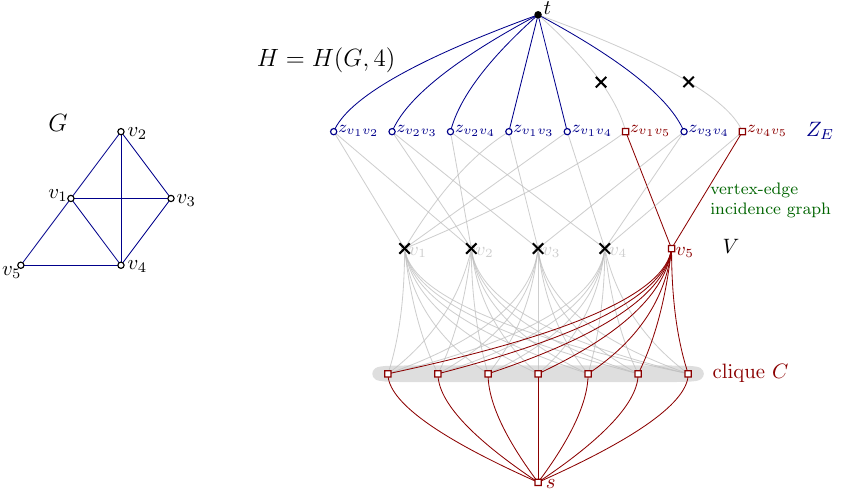}
	\caption{The $(4,2)$-cut in the graph of Figure~\ref{fig:param-by-a-1} 
		obtained from the $4$-clique $\{v_1, v_2, v_3, v_4\}$ in $G$.}
	\label{fig:param-by-a-2}
  \end{figure}
  
  We now assume that there is a solution for the \stmcu instance.
  A first observation, as $C$ has size $a+b+1$, is that one cannot disconnect $s$ from any remaining vertex of $V$ by removing vertices of $C$ and edges incident to $C$ (within their respective limit of $a$ and $b$).
  It is therefore useless to remove vertices of $C$ or edges incident to $C$.
  This also implies that the solution has to cut $\{s\} \cup C \cup V$ (or rather what is left of it) from $t$.
  Among all the mixed cuts separating $s$ from $t$ with at most $a+b$ objects in total (mixing vertices and edges), 
  we consider one using the minimum number $a' \leqslant a$ of vertices and, subject to this,
  using the minimum number $b' \leqslant b+(a-a')$ of edges.
  We next note that removing a vertex in $Z_E$ is a waste of the vertex-budget because
  instead of removing the vertex $z_e$ one can just as well remove the edge $z_et$. (Here we use the minimization of $a'$.)
  Indeed, for any edge $uv$ of $G$,
  whether the vertices $u$ and $v$ keep being connected is independent of the removal of $z_{uv}$
  because of $C$; it just depends on whether $u$ and $v$ are being removed or not.
  We can thus assume that the mixed cut \emph{only} removes vertices of $V$ and some additional edges.
  Again instead of deleting an edge $vz_e$ in the incidence graph $H[V \cup Z_E]$, 
  we can assume that we remove the edge $z_e t$.
  Indeed, the removal of $vz_{uv}$ to put $v$ and $z_{uv}$ in different components requires
  that either we remove also $uz_{uv}$ or $u$. (The vertex $t$ cannot be removed.)
  In either case we could as well remove only the edge $z_{uv}t$ (and perhaps change the connected component of $z_{uv}$).

  We have seen that we can restrict our attention to solutions that remove only vertices of $V$ and edges of $E_H(\{t\},Z_E)$.
  Let $S \subseteq V$ be the subset of $a' \leqslant a$ vertices removed by the solution  
  and note that we are removing in the solution at most $b+a-a' = m - \binom{a}{2} + a - a'$ edges of $E_H(\{t\},Z_E)$.
  Besides the edges of $E_H(\{t\},Z_E)$ that can be kept are the 
  $e(S) \leqslant \binom{a'}{2}$ corresponding to the edges in $G[S]$.
  So it holds that $\binom{a'}{2} \geqslant e(S) \geqslant \binom{a}{2} + a' - a$.
  Whenever $a=k\ge 3$, which we may assume, 
  this is only possible if $a=a'$ and $e(S)=\binom{a}{2}$, implying that $S$ is a clique of size $a=k$ in $G$. 
  
  The graph $H$ has $|V|+m+a+b+3$ vertices, can be built in polynomial time, and the parameter $a$ is set equal to $k$.
  Therefore the problem inherits the hardness of \kcli, namely W[1]-hardness and the claimed ETH lower bound.
\end{proof}

An algorithm with matching running time $n^{O(a)}$ (even $n^{a+O(1)}$) is immediate by running through all subsets $S \subset V(H)$ on up to $a$ vertices, and trying to find an edge-$(s,t)$-cut of cardinality at most $b$ on each instance $H-S$. 
In the previous reduction, we have vertices of degree three (each vertex of $Z_E$), so those vertices can be disconnected from the rest of the graph (as long as $a+b \geqslant 3$).
Therefore it does not imply any hardness for \mcu.

We use a different strategy to show that \mcu is NP-hard.
The same reduction even shows W[1]-hardness parameterized by the number of removed edges of both \mcu and its rooted version.

\begin{theorem}\label{th:param-by-b}
  \mcu and \stmcu are NP-hard and $W[1]$-hard parameterized by $b$ only.
\end{theorem}
\begin{proof}
  We reduce again from \kcli. 
  Let $G$ be the instance of \kcli. We assume without loss of generality that $k>5$.
  Set $V=V(G)$, $E=E(G)$ and $m=|E|$.
  We build $(H,a,b)$, instance of \mcu, as follows.
  See Figure~\ref{fig:param-by-b-1} for an example.
  The whole graph $H$ is partitioned into two cliques $Y \cup Z$ with $Y := V \cup Y_E \cup D_1$ and $Z := Z_E \cup D_2$,
  where both $Y_E$ and $Z_E$ are sets of vertices in one-to-one correspondence with $E$, 
  and $D_1$ and $D_2$ are two sets, each of size $a+b+1$, to force a certain structure.
  We denote by $y_e$ (resp.~$z_e$) the vertex of $Y_E$ (resp.~of $Z_E$) corresponding to the edge $e \in E(G)$.
  In addition to the edges of the cliques $Y$ and $Z$, we add the incidence graph of $G$ between $V$ and $Z_E$.
  We also add each edge $y_e z_e$ for each $e \in E(G)$; thus we have a matching between between $Y_E$ and $Z_E$.
  We set $a := m - \binom{k}{2} + k$ and $b := \binom{k}{2}$.

  \begin{figure}
  \centering
	\includegraphics[page=3,scale=.8]{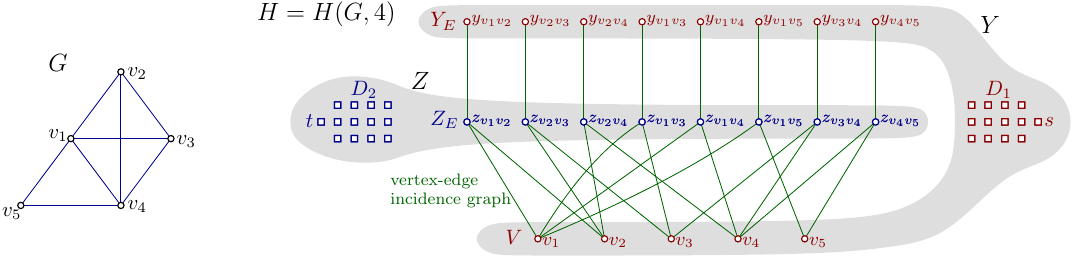}
	\caption{Example showing the reduction in the proof of Theorem~\ref{th:param-by-b}.
		On the left side we have an instance $(G,4)$ for the problem \kcli, and on the right
		we have the instance $(H,a,b)$ with $a=8-\binom{4}{2}+4=6$ and $b=\binom{4}{2}=6$ for \mcu.
		All the vertices in each of the shaded regions form a clique.}
		\label{fig:param-by-b-1}
  \end{figure}

  We now show the correctness of the reduction: the graph $G$ has a $k$-clique if and only if
  the graph $H$ has an $(a,b)$-mixed cut, under the assumption that $k>5$.
  
  Let us suppose that there is a clique $S$ of size $k$ in $G$.
  Let $Z'$ the $m - \binom{k}{2}$ vertices of $Z_E$ which do not have both endpoints in $S$.
  Let $F$ be the $\binom{k}{2}$ edges of $E_H(Z_E,Y_E)$ incident to the vertices $Z_E \setminus Z'$.
  We claim that $(S \cup Z', F)$ is an $(a,b)$-mixed cut in $H$.
  For the sizes, note that $|S \cup Z'|=|S|+|Z'| = k + m - \binom{k}{2}=a$ and $|F|=\binom{k}{2}=b$.
  Regarding the property of being a cut, since $H[V \cup Z_E]$ is the incidence graph of $G$, 
  the only edges between $Y \setminus S$ and $Z \setminus Z'$ are the $\binom{k}{2}$ edges between $Z \setminus Z'$ and $Y_E$,
  that is $F$.
  See Figure~\ref{fig:param-by-b-2} for the mixed cut that we construct for the positive instance of Figure~\ref{fig:param-by-b-1}.

  \begin{figure}
  \centering
	\includegraphics[page=4,scale=.8]{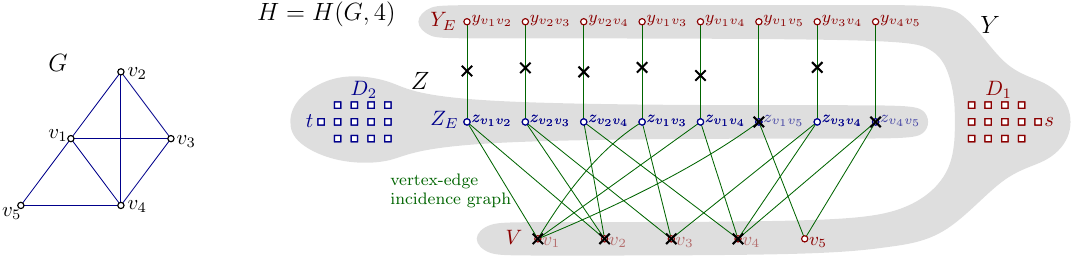}
	\caption{The $(a=6,b=6)$-cut in the graph of Figure~\ref{fig:param-by-b-1}, due to the $4$-clique $\{v_1,v_2,v_3,v_4\}$ in $G$.}
	\label{fig:param-by-b-2}
  \end{figure}

  Now let us assume that the \mcu-instance instance has an $(a,b)$-mixed cut $(W \subseteq V(H), F \subseteq E(H))$.
  Let $W_Y := W \cap Y$ and $W_Z := W \cap Z$.
  Because of the sets $D_1$ and $D_2$, $|Y| > a+b+1$ and $|Z| > a+b+1$.
  Hence it is not helpful to remove edges in the induced subgraphs $H[Y]$ or $H[Z]$, and we can assume
  that $F\subseteq E_H(Y,Z)$.  
  The problem is therefore equivalent to removing at most $a$ vertices so that there are at most $b$ edges between 
  what is left of $Y$ and what is left of $Z$.
  Since it is always better to remove the vertex $z_e$ than the vertex $y_e$, 
  one can and shall assume that $W_Y \subseteq V$ and $W_Z \subseteq Z_E$.
  
  Let us analyze $|W_Y|$ and $|W_Z|$. We will use the following property, which follows from the fact
  that the function $x\mapsto (x^2-x)/2 - 2x$ is a parabola with minimum at $x=5/2$.
  \begin{equation}\label{eq:kt}
	\forall k>5 \text{ and } k>t\ge 0:~~~~ \binom{k}{2}-2k > \binom{t}{2}-2t . 
  \end{equation}
  First note that, since $|F|=b=\binom{k}{2}$, the matching $E_H(Z_E\setminus W_Z,Y_E)$ should
  have at most $\binom{k}{2}$ edges left, which means that $Z_E\setminus W_Z$ should have at most
  $\binom{k}{2}$ vertices, and thus $|W_Z|\ge m- \binom{k}{2}$. The remaining budget of $a$ implies
  that we remove at most $k$ vertices $W_Y\subseteq V$. In short, $|W_Y|\le k$.
  
  We next show that $|W_Y|=k$.
  Assume, for the sake of reaching a contradiction, that the solution is removing $t<k$ vertices $W_Y\subseteq V$
  and $a-t=m-\binom{k}{2}+k-t$ vertices $W_Z\subset Z_E$.
  We count the remaining edges from the perspective of $Z_E\setminus W_Z$.
  To bound the edges remaining between $Z_E$ and $V$, we note that each vertex of 
  $Z_E\setminus W_Z$ has exactly two neighbors at $V$, and at most
  $|E(G[W_Y])|\le \binom{|W_Y|}{2}=\binom{t}{2}$ vertices of $Z_E\setminus W_Z$
  have both neighbors in $W_Y$. 
  Thus, each vertex of $Z_E\setminus W_Z$, but for $\binom{t}{2}$ of them,
  have at least one neighbor in $V\setminus W_Y$.
  In short, we have 
  \begin{align*}
	|E_H(V\setminus W_Y,Z_E\setminus W_Z)| ~&\ge~ 
    	|Z_E\setminus W_Z| - \binom{t}{2} 
		~=~ m - \left( m-\binom{k}{2} + k-t \right) - \binom{t}{2}\\
		~&=~ \binom{k}{2} - k + t - \binom{t}{2},
  \end{align*}
  while the number of remaining edges between $Z_E$ and $Y_E$ is at least
  \[
	|E_H(Y_E,Z_E\setminus W_Z)| ~\ge~ m - \left( m-\binom{k}{2} + k-t \right) ~=~ \binom{k}{2} - k + t   
  \]
  This means that, after the removal of $W\subseteq V\cup Z_E$, the
  number of edges between $Y$ and $Z$ that remain is 
  \[
	%|E_H(Y\setminus W_Y,Z_E\setminus W_Z)| ~=~ 
	|E_H(V\setminus W_Y,Z_E\setminus W_Z)|+|E_H(Y_E,Z_E\setminus W_Z)| ~\ge~
		2 \binom{k}{2} - 2k + 2t - \binom{t}{2} ~>~ \binom{k}{2} ~=~ b,
  \]
  where we have used \eqref{eq:kt} for the last inequality. 
  This means that removing $t<k$ vertices $W_Y\subset V$ we cannot obtain an $(a,b)$-mixed cut,
  and therefore it must be $|W_Y|=k$.
  
  From $|W_Y|=k$ and the fact that we only remove vertices in $V\cup Z_E$, we obtain
  that $|W_Z|= m- \binom{k}{2}$ and $F$ is the $\binom{k}{2}$ edges in $E_H(Z_E\setminus W_Z,Y_E)$.
  This implies that $W_Y$ is a clique in $G$ with $k$ vertices.
  
  The graph $H$ has $|V|+2(m+a+b+1)$ vertices, can be built in polynomial time, 
  and the parameter $b$ is equal to $\binom{k}{2}$.
  Therefore the problem $\mcu$ is NP-hard and inherits the W[1]-hardness of \kcli.
  The same hardness immediately holds for \stmcu by calling $s$ one vertex of $D_1$, and $t$ one vertex of $D_2$.
\end{proof}

%%%%%%%%%%%%%%%%%%%%%%%%%%%%%%%%%%%%%%%%%%%%%%%%%%%%%%%%%%%%%%%%%%%%%%%%%%%%%%%%%%%%%%%%%%%%%%%%%%%%%%%%%%%%%%%%%%%%%%%%%%%%%%%%%%%%%%%%%%%%%%%%%%%%%%%%%%%%%%%%%%%%%%%%%%%%%%%%%%%%%%%%%%%%%%%%%%%%%%%%
\section{Quadratic FPT algorithm?}
\label{sec:irrelevant}

Rai et al.~\cite{Rai16} show that \stmcu can be solved in time
$2^{O((a+b)^3\log{(a+b)})} n^4 \log n$ for graphs with $n$ vertices. Thus, the problem is 
fixed-parameter tractable in $a+b$. This implies that the \mcu is also fixed-parameter tractable, 
as we can try all $n^2$ pairs of vertices for $s$ and $t$, giving a running time
of $n^2 \cdot 2^{O((a+b)^3\log{(a+b)})} n^4 \log n= 2^{O((a+b)^3\log{(a+b)})} n^6 \log n$.
A slightly better asymptotic running time can be obtained observing that it suffices to 
take a subset $U$ of $V(G)$ with $a+1$ vertices and check the existence of
a rooted $s$-$t$ mixed separator for all the pairs
\[
	(s,t) ~\in~ \{ (u,v)\mid u\in U,~ v\in V(G),~ u\neq v\}.
\]
Indeed, if there exists an $(a,b)$-mixed separator $(W,F)$, where $W\subset V(G)$, 
$F\subset E(G)$, $|W|\le a$ and $|F|\le b$, 
then at least one of the vertices of $U$ is not in $W$ because $|U|=a+1$. 
When we try a pair $(s,t)$ with $s\in U\setminus W$ and $t$ in the component of $G-(W\cup F)$ that
does not contain $s$, we will find an $(a,b)$-mixed cut for $s$ and $t$.
(Possibly we find $(W,F)$ or another one.)
Thus, we need to invoke the algorithm Rai et al. $(a+1)(n-1)$ times, achieving a
total running time of 
$O(an) \cdot 2^{O((a+b)^3\log{(a+b)})} n^4 \log n= 2^{O((a+b)^3\log{(a+b)})} n^5 \log n$.

One of the standard approaches to try to obtain a faster FPT algorithm for \stmcu is the technique used for the 
$k$-Disjoint-Paths problem.
Kawarabayashi et al.~\cite{KKR12} show how to solve the problem in $O(n^2)$ time for any constant $k$, 
improving the previous cubic-time algorithm algorithm by Robertson and Seymour~\cite{gm13}, as part of their
graph minors project. Both papers employ the same basic structure.
In the following we show that the straightforward application of that idea does not apply here.
Some familiarity with the general structure of~\cite{gm13} or~\cite{KKR12} is convenient to follow the discussion.

The basic idea in those works is to split the algorithm into three cases: 
the graph has small treewidth, the graph has a large flat minor,
or the graph has a large clique minor. Let us concentrate in the last case: the graph $G$ has a large clique-minor.
For the sake of the discussion, we can directly assume that $G$ contains a large complete graph $K_\ell$ that is disjoint from $s$ and $t$, where $\ell \ge 3a+3b+3$ may depend on $a$ and $b$.
(Usually one would have $\ell=3a+3b+3$ or some other $\ell$ depending on $a,b$ linearly, depending on how the discussion continues.) 
The algorithm then considers two cases, depending on the minimum-size vertex cut $S$ separating $\{ s,t\}$ from some
vertex $v$ of $K_\ell$. This means that the vertex set of $G$ can be expressed as $V(G)=A\cup B$ where
$s,t\in A$, some vertex of $K_\ell$ is in $B$, there is no edge from $A\setminus B$ to $B\setminus A$,
and the size of the separator $S=A\cap B$ is minimized.
In~\cite{gm13}, the set $B$ is also chosen inclusion-wise minimal.
If the size of $S$ is large, then one can find $a+b$ vertex disjoint paths from $s$ to $t$, and thus
there is no $(a,b)$-mixed cut; see~\cite[Theorem 4.1]{KKR12}.
If the size of $S$ is small, for the disjoint paths problem, one can show that an equivalent instance
is obtained by removing $B\setminus A$ and connecting all the vertices of $S$. 
This last claim is not true for the mixed-cut. See Figure~\ref{fig:example} for an example showing 
that we can get from an instance that has a $(1,b)$-mixed cut for $s$ and $t$ and no $(1,b-1)$-mixed
cut for $s$ and $t$, but after the transformation, it has a $(1,2)$-mixed cut.
This of course does not exclude the option for faster algorithms modifying the approach.
\begin{figure}[h!]
\centering
	\includegraphics[page=6,scale=.7]{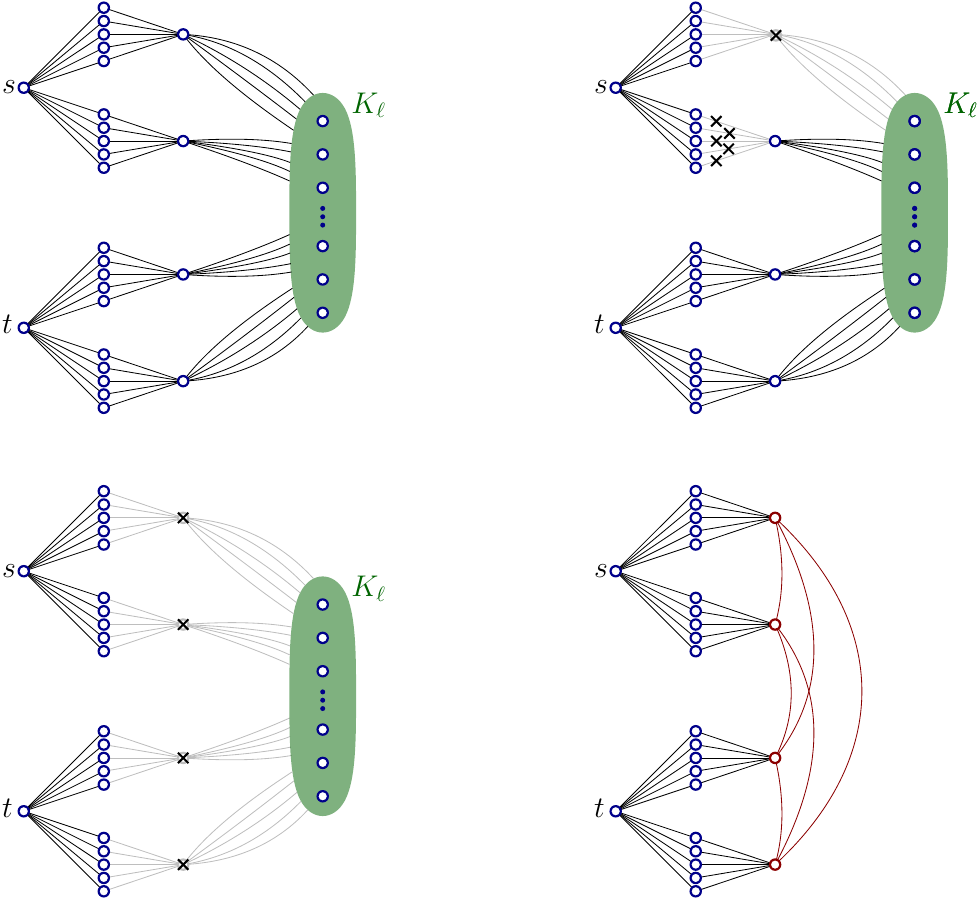}
	\caption{An instance for \stmcu (top left) that has an $(a=1,b=5)$-mixed cut for $s$ and $t$ (top right), 
		but no $(1,4)$-mixed cut for $s$ and $t$.
		The instance has an arbitrary large clique $K_\ell$ and there is a vertex-separation between $\{ s,t \}$
		and the clique $K_\ell$ with four vertices (bottom left). 
		Removing the part of the separation that contains $K_\ell$ and connecting the vertices of the
		four-vertex separator (bottom right), we get an instance that has a $(1,2)$-mixed cut for $s$ and $t$.
		The example can be easily generalized to any $b>5$.}
	\label{fig:example}
\end{figure}

%\bibliographystyle{abuser}
%\bibliography{main}

\end{document}